\documentclass{article}
\usepackage[nohyperref, accepted]{icml2013}

\usepackage{color}
\usepackage{url}
\usepackage{verbatim} 
\usepackage{subfigure} 
\usepackage{multirow}
\usepackage{algorithm}
\usepackage{algorithmic}
\usepackage{xspace}
\usepackage{graphicx}
\usepackage{epsfig}
\usepackage{amsmath}
\usepackage{amssymb}
\usepackage{amsthm}
\usepackage{natbib}
\usepackage{hyperref}
\usepackage{ctable}

\newcommand{\hide}[1]{}
\newcommand{\xhdr}[1]{\vspace{1.7mm}\noindent{{\bf #1.}}}
\newcommand{\casc}{{\mathbf{t}}}
\newcommand{\alphs}{{\mathbf{A}}}

\newtheorem{theorem}{Theorem}

\newcommand{\netrate}{{\textsc{Net\-Rate}}\xspace}
\newcommand{\infopath}{{\textsc{Info\-Path}}\xspace}
\newcommand{\monet}{{\textsc{mo\-Net}}\xspace}
\newcommand{\kernelcascade}{{\textsc{Kernel\-Cascade}}\xspace}

\newcommand{\expo}{{\textsc{Exp}}\xspace}
\newcommand{\pow}{{\textsc{Pow}}\xspace}
\newcommand{\ray}{{\textsc{Ray}}\xspace}

\newcommand{\ie}{\emph{i.e.}}
\newcommand{\unobs}{{\infty}}

\graphicspath{{./FIG/}}

\begin{document}

\icmltitlerunning{Modeling Information Propagation with Survival Theory}

\twocolumn[
\icmltitle{Modeling Information Propagation with Survival Theory}

\icmlauthor{Manuel Gomez-Rodriguez$^{1,2}$}{manuelgr@tuebingen.mpg.de}
\icmlauthor{Jure Leskovec$^2$}{jure@cs.stanford.edu}
\icmlauthor{Bernhard Sch\"{o}lkopf$^1$}{bs@tuebingen.mpg.de}

\icmladdress{$^1$MPI for Intelligent Systems and $^2$Stanford University}

\icmlkeywords{influence maximization, diffusion networks, social networks, temporal dynamics}

\vskip 0.3in
]

\begin{abstract}
Networks provide a `skeleton' for the spread of contagions, like, information, ideas, be\-haviors and diseases. Many times networks over which contagions diffuse are unobserved and need to be inferred.
Here we apply survival theory to develop ge\-ne\-ral additive and multiplicative risk models under which the network inference pro\-blems can be solved effi\-cient\-ly by exploiting their convexity. Our addi\-tive 
risk model generalizes se\-ve\-ral exis\-ting network inference models. We show all these models are particular cases of our more ge\-ne\-ral model. Our multiplicative model allows\- for modeling scenarios in 
which a node can either increase or decrease the risk of activation of another node, in contrast with previous approaches, which consider only positive risk increments. We evaluate the performance of our 
network inference algorithms on large synthetic and real cascade datasets, and show that our models are able to predict the length and duration of cascades in real data.
\end{abstract}

\section{Introduction}
\label{sec:intro}

Network diffusion is one of the fundamental processes ta\-king place in networks~\cite{rogers95diffusion}. For example, information, diseases, rumors, and behaviors spread over underlying social and information networks.
Abs\-trac\-tly, we think of a {\em contagion} that appears at some node of a network and then spreads like an epidemic from node to node over the edges of the network. For example, in information 
propagation, the contagion co\-rres\-ponds to a piece of information~\cite{nowell08letter, leskovec2009kdd}, the nodes corres\-pond to people and infection events are the times when nodes learn about the information. Similarly, we can think about the spread of a new type of be\-havior or an action, like, purchasing and re\-co\-mmen\-ding a new product~\cite{jure06viral} or the propagation of a contagious disease over social network of in\-di\-vi\-duals~\cite{bailey75mathematical}.

Propagation often occurs over networks which are hi\-dden or un\-ob\-served. However, we can observe the trace of the contagion spreading. For example, in information diffusion, we observe when a node learns about the information but not who they heard it from. In epidemiology, a person can become ill but cannot tell who infected her. And, in marketing, it is possible to observe when customers buy products but not who influenced their decisions. Thus, we can observe a set of contagion infection times and the goal is to infer the edges of the underlying network over which the contagion diffused~\cite{manuel10netinf}.

In this paper we propose a general theoretical framework to model information propagation and then infer hidden or unobserved networks using survival theory. We generalize previous work, develop efficient network inference methods, and validate them experimentally. In particular, our methods not only identify the network structure but also infer which links inhibit or encourage the diffusion of the contagion.

\xhdr{Our approach to information propagation} We consider {\em contagions} spreading across a fixed population of nodes. The contagion spreads by nodes forcing other nodes to switch from being uninfected to being infected, but nodes cannot switch in the opposite direction. Therefore, we can represent whether a node is infected at any given time as a nondecreasing (binary) counting process. 
We then model the instantaneous risk of infection, i.e., {\em the hazard rate}~\cite{aalen2008survival} of a node by using the infection times of other previously infected nodes as explanatory va\-ria\-bles or \emph{covariates}. 
By inferring which nodes in\-fluence the ha\-zard rate of a given node, we discover the edges of the underlying network over which propagation takes place. In particular, if the hazard rate of node $i$ depends on the infection time of node $j$, then there is a directed edge $(j, i)$ in the underlying network.

We then develop two models. First, we introduce an additive risk model under which the hazard rate of each node is an additive function of the infection times of other previously infected nodes. We show that se\-ve\-ral previous approaches to network inference~\cite{manuel11icml, lesong2012nips, wang2012feature, manuel13dynamic} are particular cases of our more general additive risk model. How\-ever, all these models im\-pli\-cit\-ly consider previously infected nodes to only {\em increase} the instantaneous risk of infection.
We then relax this assumption and develop a multiplicative risk model under which the ha\-zard rate of each node is multiplicative on the infection times of other previously infected nodes. This allows previously infected nodes to 
either {\em increase} or {\em decrease} the risk of another node getting infected.
For example, trendsetters'{} probability of buying a product may increase when she observes her peers buying a product but may also {\em decrease} when she realizes that \emph{average, mainstream} friends are 
buying the product. Similarly, consider an example of a blog which often mentions pieces of information from a general news media sites, but only whenever they are not related to sports. Therefore, if the general news 
media site publishes a piece of information related to sports, we would like the blog'{}s risk of adopting the information to be smaller than for other type of information. Last, we show how to efficiently fit the parameters of 
both models by using the maximum likelihood principle and by exploiting convexity of the optimization problems.

\xhdr{Related work} In recent years, many network inference algorithms have been developed~\cite{saito2009learning, manuel10netinf, manuel11icml, manuel13dynamic, meyers10netinf, 
snowsill2011refining, netrapalli12, multitree12icml, wang2012feature}. 
These approaches differ in a sense that some infer only the network structure~\cite{manuel10netinf, snowsill2011refining}, while others infer not only the network structure but also the \emph{strength} or the average 
latency of every edge in the network~\cite{saito2009learning, meyers10netinf, manuel11icml, manuel13dynamic, wang2012feature}. Most of the approaches use only temporal information while a few 
methods~\cite{netrapalli12, wang2012feature} consider both temporal information and additional non-temporal features. Our work provides two novel contributions over above approa\-ches. First, our additive risk model is a ge\-ne\-ra\-li\-zation of several models which have been proposed previously in the literature. Second, we develop a multiplicative model which allows nodes to increase or decrease the risk of infection of another node.

\section{Modeling information propagation with survival analysis}
\label{sec:formulation}

\xhdr{Information propagation as a counting process}
We consider multiple independent contagions spreading\- across an unobserved network on $N$ nodes. 
As a single contagion spreads, it creates a {\em cascade}. 
A cascade $\casc^c$ of contagion $c$ is simply a $N$-dimensional vec\-tor $\casc^c:=(t^c_1,\ldots,t^c_N)$ re\-cor\-ding the times when each of $N$ nodes got infected by the contagion $c$: $t^c_i\in [t_{0},t_{0}+T^c]\cup\{\unobs\}$, where $t_{0}$ is the infection time of the first node. 
Generally, contagions do not infect all the nodes of the network, and symbol $\unobs$ is used for nodes that were not infected by the contagion $c$ during the observation window $[t_{0},t_{0}+T^c]$. 
For simplicity, we assume $T^c = T$ for all cascades; the results generalize trivially.
In an information or rumor propagation se\-tting, each cascade $c$ corresponds to a different piece of information or rumor, nodes $i$ are people, and the infection time of a node $t_i^c$ is simply the time when node $i$ first 
learned about the piece of information or rumor.

Now, consider node $i$, cascade $\casc^c = \casc$, and an indicator function $N_{i}(t)$ such that $N_{i}(t) = 1$ if node $i$ is infected by time $t$ in the cascade and $N_{i}(t) = 0$ otherwise. Then, we define the {\em filtration} $\mathcal{F}_{t}$ as the set of nodes that has been infected by time $t$ and their infection times, \ie, $\mathcal{F}_{t} = (\mathbf{t}^{<t})$,
where $\mathbf{t}^{<t} = (t_1, \ldots, t_j, \ldots, t_N | t_j < t)$. By definition, since $N_{i}(t)$ is a nondecreasing counting process, it is a submartingale and satisfies that $E(N_{i}(t) | \mathcal{F}_{t'}) \geq N_{i}(t')$ for any $t > t'$. 
Then we can decompose $N_{i}(t)$ \emph{uniquely} as $N_{i}(t) = \Lambda_{i}(t) + M_{i}(t)$, where $\Lambda_{i}(t)$ is a nondecreasing predictable process, called 
\emph{cumulative intensity process} and $M_{i}(t)$ is a mean zero martingale. This is called the Doob-Meyer decomposition of a submartingale~\cite{aalen2008survival}. 
Consider $\Lambda_{i}(t)$ to be absolutely con\-ti\-nuous, then there exists a \emph{predictable} nonnegative intensity process $\lambda_{i}(t)$ such that:
\begin{equation} \label{eq:doob-meyer-micro}
N_{i}(t) = \int_0^{t} \lambda_{i}(s)\,ds + M_{i}(t).
\end{equation}
Now, we assume that the intensity process $\lambda_{i}(t)$ depends on a vector of explanatory variables or \emph{co\-va\-riates}, $\mathbf{s}(t) = \gamma(\mathbf{t}^{<t} ; t)$, where $\gamma(\cdot)$ is an arbitrary time shaping function that we have to decide upon. In our case the covariate vector accounts for the previously infected nodes up to the time just before $t$, \ie, the filtration $\mathcal{F}_{t}$. 
Then, we can rewrite the intensity process of $N_{i}(t)$ as $\lambda_{i}(t) = Y_{i}(t) \alpha_{i}(t | \mathbf{s}(t))$, where $Y_{i}(t)$ is an indicator such that $Y_{i}(t) = 1$ if node $i$ is susceptible to 
be infected just before time $t$ and $0$ otherwise, and $\alpha_{i}(t | \mathbf{s}(t))$ is called the intensity or hazard rate of node $i$ and it is defined conditional on the values of the co\-va\-riates. Note that 
the hazard rate must be nonnegative at any time $t_i$ since otherwise $N_{i}(t)$ would decrease, violating the assumptions of our framework. In other words, a node remains susceptible as long as it did not get infected.

Our goal now is to infer the hazard function $\alpha_{i}(t | \mathbf{s})$ for each node $i$ from a set of recorded cascades $\mathbf{C}= \{\mathbf{t}^1, \ldots, \mathbf{t}^{|C|}\}$. 
This will allow us to discover the edges of the underlying network and also predict future infections. 
In particular, if there is an edge $(j, i)$ in the underlying network, the hazard rate of node $i$ will depend on the infection time of node $j$. Therefore, the ha\-zard $\alpha_{i}(t | \mathbf{s})$ tells us about the incoming edges to node $i$.
Also, we will be able to predict future infections by computing the cumulative probability $F_{i}(t | \mathbf{s}(t_i))$ of infection of a susceptible node $i$ at any given time $t$ using the hazard 
function~\cite{aalen2008survival}:
\begin{equation} \label{eq:cdf}
F_{i}(t | \mathbf{s}(t)) = 1 - e^{-\int_0^{t} \alpha_{i}(t' | \mathbf{s}(t'))\,dt'}.
\end{equation} 
%
In the remainder of the paper, we propose an additive and a multiplicative model of hazard functions $\alpha_{i}(t | \mathbf{s})$ and validate them experimentally in synthetic and real data. There are 
several reasons to do so. First, to provide a general framework which is flexible enough to fit cascading processes over networks in different domains. Second, to allow for both positive and negative influence 
of a node in its neighbors'{} hazard rate, wi\-thout violating the nonnegativity of hazard rates over time. Third, as it has been argued the necessity of both additive and multiplicative models in traditional survival 
analysis literature~\cite{aalen2008survival}, our framework also supports networks in which some nodes have additive hazard functions while others have multiplicative hazard functions.

\section{Additive risk model of information propagation}
\label{sec:additive}
%
First we consider the hazard function $\alpha_{i}(t | \mathbf{s}(t))$ of node $i$ to be additive on the infection times of other pre\-vious\-ly infected nodes. We then show that this model is equivalent to the continuous 
time independent cascade model~\cite{manuel11icml}.

Consider the hazard rate of node $i$ to be: 
\begin{equation} \label{eq:hazard-additive}
\alpha_i(t | \mathbf{s}(t)) = \boldsymbol{\alpha}_i^T  \mathbf{s}(t)=\boldsymbol{\alpha}_i^T \gamma(\mathbf{t}^{<t} ; t),
\end{equation}
where $\boldsymbol{\alpha}_i = (\alpha_{1i}, \ldots, \alpha_{Ni}) \geq \mathbf{0}$ is a nonnegative parameter vector and $\gamma(\cdot) \geq \mathbf{0}$ is an arbitrary positive time shaping function on the previously 
infected nodes up to time $t$. We force the parameter vector and time shaping function to be nonnegative to avoid ill-defined negative hazard functions at any time $t$.
We then assume\- that each covariate depends only on one previously infected node and therefore each parameter $\alpha_{ji}$ only models the effect of a single node $j$ on node $i$. Then, 
$\boldsymbol{\alpha}_i \in \mathbb{R}^{N}_{+}$, where $\alpha_{ii} = 0$.
For simplicity, we apply\- the same time shaping function to each of the parents'{} infection times, where we define parents of node $i$ to be a set of nodes $j$ that point to $i$. Mathematically this means that $\gamma(\mathbf{t}^{<t} ; t) = (\gamma(t_1 ; t), \ldots, \gamma(t_{N} ; t))$. 

Our goal now is to infer the optimal parameters $\boldsymbol{\alpha}_i$ for every node $i$ that maximize the likelihood of a set of observed cascades $C$. By inferring the parameter vector $\boldsymbol{\alpha}_i$, we also discover the underlying network over which propagation occurs: if $\alpha_{ji} \neq 0$, there is an edge $(j, i)$, and if $\alpha_{ji} = 0$, then there is no edge. 

We proceed as follows. First we compute the cumulative likelihood $F_i(t | \mathbf{s}(t))$ of infection of node $i$ from the hazard rate using Eq.~\ref{eq:cdf}:
\begin{equation} \label{eq:cdf-additive}
F_i(t | \mathbf{s}(t); \boldsymbol{\alpha}_i) = 1 - \prod_{j : t_j < t} e^{- \alpha_{ji} \int_{t_j}^{t} \gamma(t_j ; t')\,dt' }.
\end{equation}
Then, the likelihood of infection $f_i(t | \mathbf{s}(t))$ is:
\begin{equation} \label{eq:pdf-additive}
f_i(t | \mathbf{s}(t); \boldsymbol{\alpha}_i) = \sum_{j : t_j < t} \alpha_{ji} \gamma(t_j ; t) \prod_{k : t_k < t} e^{-\alpha_{ki} \int_{t_k}^{t} \gamma(t_k ; t')\,dt'}.
\end{equation}
Now, consider cascade $\casc:=(t_1, \ldots, t_N)$. We first compute the likelihood of the observed infection times $\casc^{\leq T} = (t_1, \ldots, t_N | t_i\leq T)$ during the observation window $T$. Each infection is con\-di\-tio\-na\-lly independent on infections which occur later in time given previous infections. Then, the likelihood factorizes over nodes as:
\begin{multline} \label{eq:likelihood-cascade-infections-additive}
f(\casc^{\leq T}; \alphs) = \prod_{i : t_i < T} \sum_{j : t_j < t_i} \alpha_{ji} \gamma(t_j ; t_i) \times \\
\prod_{k : t_k < t_i} e^{- \alpha_{ki} \int_{t_k}^{t_i} \gamma(t_k ; t)\,dt}.
\end{multline}
where $\alphs := \left[ \boldsymbol{\alpha_i} \right] \in \mathbb{R}^{N\times N}_{+}$. However, Eq.~\ref{eq:likelihood-cascade-infections-additive} only considers infected nodes. The fact that some nodes are \emph{not} infected during the observation 
window is also informative. We thus add survival terms 
for any noninfected node $n$ ($t_n > T$, or equivalently $t_n=\infty$) and apply logarithms. Therefore, the log-likelihood of cascade $\mathbf{t}$ is:
\begin{equation} \label{eq:log-likelihood-cascade-additive}
\begin{split}
\log f(\mathbf{t} ; \alphs) &= \sum_{i : t_i < T} \log \left(\sum_{j : t_j < t_i} \alpha_{ji} \gamma(t_j ; t_i)\right) \\
&- \sum_{i : t_i < T} \sum_{k : t_k < t_i} \alpha_{ki} \int_{t_k}^{t_i} \gamma(t_k ; t)\,dt\\
&- \sum_{n : t_n > T} \sum_{m : t_m < T} \alpha_{m,n} \int_{t_m}^{T} \gamma(t_m ; t)\,dt,
\end{split}
\end{equation}
where the first two terms represent the infected nodes, and the third term represents the noninfected nodes at the end of the observation window $T$. 
As each cascade propagates independently of others, the log-likelihood of a set of cascades $C$ is the sum of the log-likelihoods of the individual cascades. 
Now, we apply the maximum likelihood principle on the log-likelihood of the set of cascades to find the optimal parameters $\boldsymbol{\alpha}_i$ of every node $i$:
\begin{equation}
	\label{eq:opt-problem-additive}
	\begin{array}{ll}
		\mbox{minimize$_{\alphs}$} & - \sum_{c \in C} \log f(\casc^c;\alphs) \\
		\mbox{subject to} & \alpha_{ji} \geq 0,\, i,j=1,\ldots,N, i \neq j,
	\end{array}
\end{equation}
where $\alphs := \left[ \boldsymbol{\alpha_i} \right] \in \mathbb{R}^{N\times N}_{+}$. The solution to Eq.~\ref{eq:opt-problem-additive} is unique and computable:
\begin{theorem}
	The network inference problem for the additive risk model defined in Eq.~\ref{eq:opt-problem-additive} is convex in $\alphs$.
\end{theorem}
\begin{proof} 
Convexity follows from linearity, composition rules for convexity, and concavity of the lo\-ga\-rithm.
\end{proof}
%
%
There are several common features of the solutions to the network inference problem under the additive risk model. The first term in the log-likelihood of each cascade, defined by Eq.~\ref{eq:log-likelihood-cascade-additive}, 
ensures that for each infected node $i$, there is at least one previously infected parent since otherwise the log-likelihood would be negatively unbounded, \ie, $\log 0 = 1$. Moreover, there 
exists a natural diminishing property on the number of pa\-rents of a node -- since the logarithm grows slowly, it weakly rewards infected nodes for having many parents.
The second and third term in the log-likelihood of each cascade, defined by Eq.~\ref{eq:log-likelihood-cascade-additive}, consist of positively weighted L1-norm on the vector $\alphs$. L1-norms are well-known heuristics to encourage sparse solutions~\cite{boyd2004convex}. That means, optimal networks under the additive risk model are sparse.
\begin{table}[t]
    \begin{center}
    \begin{small}
    \begin{tabular*}{0.48\textwidth}{@{\extracolsep{\fill}} l l}
        \toprule
	\textbf{Network Inference Method} &  $\bf \gamma(t_j ; t_i)$ \\
	\midrule
	\netrate, \infopath (\expo) & $I(t_j < t_i)$  \\
	\netrate, \infopath (\pow) & $\max(0, 1/(t_i - t_j))$ \\
	\netrate, \infopath (\ray) & $\max(0, t_i - t_j)$ \\
	\kernelcascade & $\{k(\tau_l,  t_i - t_j)\}_{1}^{m}$ \\
	\monet & $I(t_j < t_i) \gamma e^{-d(\mathbf{f}_j, \mathbf{f}_i)}$ \\
	\bottomrule
    \end{tabular*}
    \end{small}
    \end{center}
    \vspace{-4mm}
    \caption{Mapping from several network inference methods to our general additive risk model.}
    \label{tab:mapping}
\end{table}

\xhdr{Generalizing present network inference methods} 
Network in\-fe\-rence methods~\cite{manuel11icml, manuel13dynamic, lesong2012nips, wang2012feature} model information propagation using continuous time generative probabilistic models of di\-ffu\-sion. In 
such models, one typically starts describing the pairwise interactions between pairs of nodes. One defines a pairwise infection likelihood $f_i(t | t_j ; \theta_{ji})$ of node $j$ infecting node $i$. Then, one 
continues computing the likelihood of an infection of a node by assuming a node gets infected once any of the previously infected nodes succeeds\- at infecting her, as in the independent cascade 
model~\cite{kempe03maximizing}. As a final step, the likelihood of a cascade is computed from the likelihoods of individual infections. The network inference problem can then be solved by finding the network that maximizes the likelihood of observed infections. Importantly, the following result holds:
\begin{theorem} \label{th:mapping}
The continuous time independent cascade model~\cite{manuel11icml} is an a\-dditive hazard model on the pairwise hazards between a node and her parents.
\end{theorem}
\vspace{-3mm}
\begin{proof}
In the continuous time independent cascade model, for a given node $i$, the likelihood of infection 
$f_i(t | \mathbf{t}^{<t}; \Theta)$ and the probability of survival $S_i(t | \mathbf{t}^{<t} ; \Theta)$ given the previously infected nodes $\mathbf{t}^{<t}$ are:
\begin{align*}
f(t | \mathbf{t}^{<t} ; \Theta) &= \prod_{k : t_k < t} S(t | t_k ; \theta_{ki}) \sum_{j : t_j < t} \alpha_i(t | t_j ; \theta_{ji}), \\
S(t | \mathbf{t}^{<t} ; \Theta) &= \prod_{k : t_k < t} S(t | t_k ; \theta_{ki}).
\end{align*}
Then, the hazard of node $i$ is
\begin{equation}
\alpha_{i}(t | \mathbf{t}^{<t} ; \Theta) = \sum_{j : t_j < t} \alpha_i(t | t_j ; \theta_{ji}),
\end{equation}
which is trivially additive on the pairwise hazards between a node and her parents.
\end{proof}
Therefore, our model is a generalization of the con\-ti\-nuous time independent cascade model. Several other models used by state of the art network
inference methods map easily to our general additive risk model (see Table~\ref{tab:mapping}). For example, pairwise transmission likelihoods used in \netrate~\cite{manuel11icml} and 
\infopath~\cite{manuel13dynamic} result in simple pairwise hazard rates that map into our model by setting the time shaping functions $\gamma(\cdot)$. 
The \emph{kernelized} hazard functions used in \kernelcascade~\cite{lesong2012nips} map into our model by considering $m$ covariates per parent, 
where $k(\tau_l, \cdot)$ is a kernel function and $\tau_l$ is the $l^{th}$ point in a $m$-point uniform grid over of $[0, T]$. This allows to model multimodal
hazard functions.
Finally, the featured-enhanced diffusion model used in \monet~\cite{wang2012feature} maps into our model by considering a time shaping function with 
both temporal and non-temporal covariates, where $d(\mathbf{f}_j, \mathbf{f}_i)$ denote the distance between two non-temporal feature vectors and $\gamma$ 
is the normalization constant.

\section{Multiplicative risk model of information propagation}
\label{sec:multiplicative}
Existing approaches to network inference only consider edges in the network to increase the hazard rate of a node. We next provide an extension where we can model
situations in which a parent can either increase or decrease the hazard rate of the target node. We achieve this by examining a case where the hazard function $\alpha_{i}(t | \mathbf{s}(t))$ of 
node $i$ is multiplicative on the covariates, i.e., infection times of other nodes of the network. We consider the hazard rate of node $i$ to be: 
\begin{equation} \label{eq:hazard-multiplicative}
\alpha_i(t | \mathbf{s}(t)) = \alpha_{0i}(t) \prod_{j : t_j < t} \beta_{ji},
\end{equation}
where $\alpha_{0i}(t) \geq 0$ is a fixed or time varying baseline function, which is independent of the previously infected nodes, and $\beta_{ji} \geq \epsilon > 0$ are the parameters of the model, which represent the \emph{positive} or \emph{negative} influence of node $j$ on node $i$. If $\beta_{ji} > 1$, then when node $j$  gets infected, the instantaneous risk of infection of node $i$ increases. Similarly, if $\beta_{ji} < 1$, then it 
decreases, and, if $\beta_{ji} = 1$, node $j$ does not have any effect on the risk of node $i$, i.e., there is no edge in the network. 
The baseline function $\alpha_{0i}(t)$ have a complex shape and is chosen based on expert knowledge. For simplicity, we consider simple functions such as 
$\alpha_{0i}(t)=e^{\alpha_{0i}}$, $\alpha_{0i}(t)=e^{\alpha_{0i}} t$, or $\alpha_{0i}(t)=e^{\alpha_{0i}} / t$, where we set $\alpha_{0i}$ to some value equal for all nodes $i$. We note that we also tried to include $\alpha_{0i}$ as a variable
in the network inference problem, but this did not lead to improved performance.

Our goal now is to infer the optimal parameters $\beta_{ji}$ that maximize the likelihood of a set of observed cascades $C$. Importantly, by inferring the parameters $\beta_{ji}$, we also discover the underlying network over which propagation occurs. If $\beta_{ji} \neq 1$, then there is an edge from node $j$ to node $i$, and if $\beta_{ji} = 1$, there is not edge. 

To this aim, we need to compute the likelihood of a cascade starting from the hazard rate of each node.
We first compute the cumulative likelihood $F_i(t | \mathbf{s}(t))$ of infection of a node $i$ using Eq.~\ref{eq:cdf}:
\begin{multline} \label{eq:cdf-multiplicative}
F_i(t | \mathbf{s}(t);  \boldsymbol{\beta_{i}}) = \\
1 - \exp \left(- \sum_{j : t_j \leq t, j > 0} \prod_{k : t_k < t_j, k > 0} \beta_{ki} \int_{t_{j-1}}^{t_j} \alpha_{0i}(t) \,dt\right),
\end{multline}
where $\boldsymbol{\beta}_i = (\beta_{1i}, \ldots, \beta_{Ni})$. Then, the likelihood of infection $f_i(t | \mathbf{s}(t))$ is:
\begin{equation} \label{eq:pdf-multiplicative}
f_i(t | \mathbf{s}(t); \boldsymbol{\beta_{i}}) = \alpha_{0i}(t) \left(\prod_{k : t_k < t} \beta_{ki}\right) (1-F_i(t | \mathbf{s}(t))),
\end{equation}
where the indices indicate temporal order, $t_0 = 0 < t_1 < \ldots < t_{i-1} < t$. The key observation to compute the likelihood of infection from the cumulative likelihood is to realize that there is only one integral in the cumulative likelihood that contains $t$, and so we only need to take the derivative with respect to variable $t$. 

Now, consider cascade $\casc:=(t_1, \ldots, t_N)$. We first compute the likelihood of the observed infections $\casc^{\leq T}=(t_1,\ldots,t_N|t_i\leq T)$. Each infection is con\-di\-tio\-na\-lly independent on infections which occur later in time given previous infections. Then, the likelihood fac\-to\-rizes over the nodes as:
%
%
%
\begin{multline} \label{eq:likelihood-cascade-infections-multiplicative}
f(\casc^{\leq T}; B) = \prod_{i : t_i < T} \alpha_{0i}(t_i) \prod_{k : t_k < t_i} \beta_{ki} \times \\
\exp \left( \sum_{j : t_j \leq t_i} \prod_{k : t_k < t_j, k > 0} \beta_{ki} \int_{t_{j-1}}^{t_j} \alpha_{0i}(t) \,dt \right),
\end{multline}
where $B := \{\beta_{ji}\,|\, i,j=1,\ldots,n, i \neq j\}$. However, Eq.~\ref{eq:likelihood-cascade-infections-multiplicative} only considers infected nodes. The fact that some nodes are \emph{not} infected by the contagion
is also informative. We then add survival terms 
for any noninfected node $n$ ($t_n > T$, or equivalently $t_n=\infty$). We now reparameterize $\beta_{ji}$ to $\alpha_{ji} = \log(\beta_{ji})$ and apply logarithms to compute the log-likelihood of a cascade as,
\begin{multline} \label{eq:log-likelihood-cascade-multiplicative}
\log f(\mathbf{t} ; \alphs) = \sum_{i : t_i < T} \sum_{k : t_k < t_i} \alpha_{ki} + \sum_{i : t_i < T} \log(\alpha_{0i}(t_i)) \\
- \sum_{i : t_i < T} \sum_{j : t_j \leq t_i} e^{\sum_{k : t_k < t_j, k > 0} \alpha_{ki}} \int_{t_{j-1}}^{t_j} \alpha_{0i}(t) \,dt \\
- \sum_{n : t_n > T} \sum_{j : t_j \leq T} e^{\sum_{k : t_k < t_j, k > 0} \alpha_{ki}} \int_{t_{j-1}}^{t_j} \alpha_{0i}(t) \,dt,
\end{multline}
where $\alphs := \left[ \boldsymbol{\alpha_i} \right] \in \mathbb{R}^{N\times N}$ and $\alpha_{ii} = 0$. The first three terms represent the infected nodes and the last term represents the surviving ones up to the observation 
window cut-off $T$. 
%
Assuming independent cascades, the log-likelihood of a set of cascades $C$ is the sum of the log-likelihoods of the individual cascades given by Eq.~\ref{eq:log-likelihood-cascade-multiplicative}. Then, we apply the 
maximum likelihood principle on the log-likelihood of the set of cascades to find the optimal parameters $\boldsymbol{\alpha}_i$ of every node $i$:
\begin{equation}
	\label{eq:opt-problem-multiplicative}
	\begin{array}{ll}
		\mbox{minimize$_{\alphs}$} & - \sum_{c \in C} \log f(\casc^c;\alphs)
	\end{array}
\end{equation}
%
The solution to Eq.~\ref{eq:opt-problem-multiplicative} is unique and computable:
\begin{theorem} \label{th:network-inference-multiplicative}
	The network inference problem under the multiplicative risk model defined in Eq.~\ref{eq:opt-problem-multiplicative} is convex in $\alphs$.
\end{theorem}
\begin{proof} 
Result follows from linearity, composition rules for convexity, and convexity of the exponential.
\end{proof}

Model parameters have natural interpretation. If $\alpha_{ji} > 0$, node $j$ increases the hazard rate of node $i$ (positive influence), if $\alpha_{ji} < 0$, node $j$ decreases the hazard rate of node $i$ (negative influence), 
and finally if a parameter $\alpha_{ji} = 0$, node $j$ does not have any influence on $i$ -- there is no edge between $j$ and $i$. 

However, there are some undesirable properties of the solution to the multiplicative risk model as defined 
by Eq.~\ref{eq:opt-problem-multiplicative}. The optimal network will be dense: any pair of nodes $(j,i)$ that are not infected by the same contagion at least once will have negative influence on each other. Even worse, the 
negative influence between those pairs of nodes will be arbitrarily large, making the optimal solution unbounded. 
We propose the following solution to this issue. If pair $(j, i)$ does not get infected in any common cascades, we set $\alpha_{ji}$ to zero and do not include it in the log-likelihood computation. This rules out interactions between nodes that got infected in disjoint sets of cascades and avoids unbounded optimal solutions. In other words, we assume that if node $j$ has (positive or negative) influence on node $i$, then $i$ and $j$ should get infected by at least one common contagion and naturally $j$ should get infected before $i$. By ruling out interactions between nodes that got infected in disjoint cascades we successfully reduce the network density of the optimal solution. However, the solution is not encouraged to be sparse yet. We achieve even greater sparsity by 
including L1-norm regularization term~\cite{boyd2004convex}. 
Therefore, we finally solve:
\begin{equation}
	\label{eq:opt-problem-multiplicative-positive}
	\begin{array}{ll}
		\mbox{minimize$_{\alphs}$} & - \sum_{c \in C} g(\casc^c;\alphs) + \lambda \sum_{j, i} |\alpha_{ji}|,
	\end{array}
\end{equation}
where $\lambda$ is a sparsity penalty parameter and $g(\casc^c;\alphs)$ is the log-likelihood of cascade $\casc^c$ which omits pa\-ra\-meters $\alpha_{ji}$ of pairs $(j,i)$ that did not get infected by at least one common contagion 
(and $t_j < t_i$). The above problem is convex by using the same reasoning as in Th.~\ref{th:network-inference-multiplicative}.
Finally, we note that by introducing a L1-norm regularization term, we are essentially assuming Laplacian prior over $\alphs$. Depending on the domain, other priors may be more appropriate; as long as they are jointly log-concave on $\alphs$, the network inference problem will still be convex.

\section{Experimental evaluation}
\label{sec:evaluation}
We evaluate the performance of both the additive and the multiplicative model on synthetic networks that mimic the structure of real networks as well as on a dataset of more than 10 million information cascades spreading 
between 3.3 million websites over a 4 month period\footnote{Available at http://snap.stanford.edu/infopath/}.
\begin{figure}[t]
	\centering
	\subfigure[Edge accuracy (C-P)]{\includegraphics[width=0.22\textwidth]{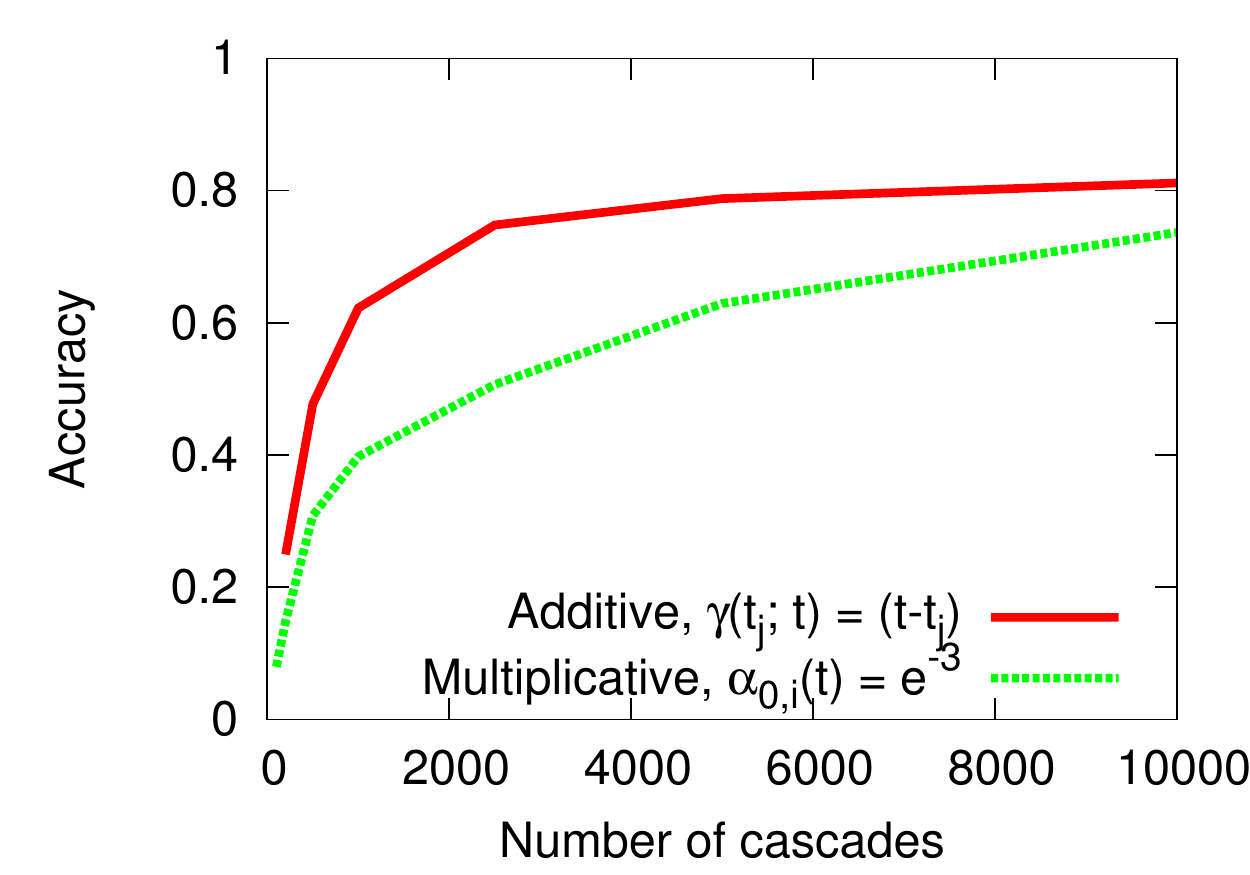}
	\label{fig:acc-cp-vs-num-cascades}}
	\subfigure[Edge accuracy (HI)]{\includegraphics[width=0.22\textwidth]{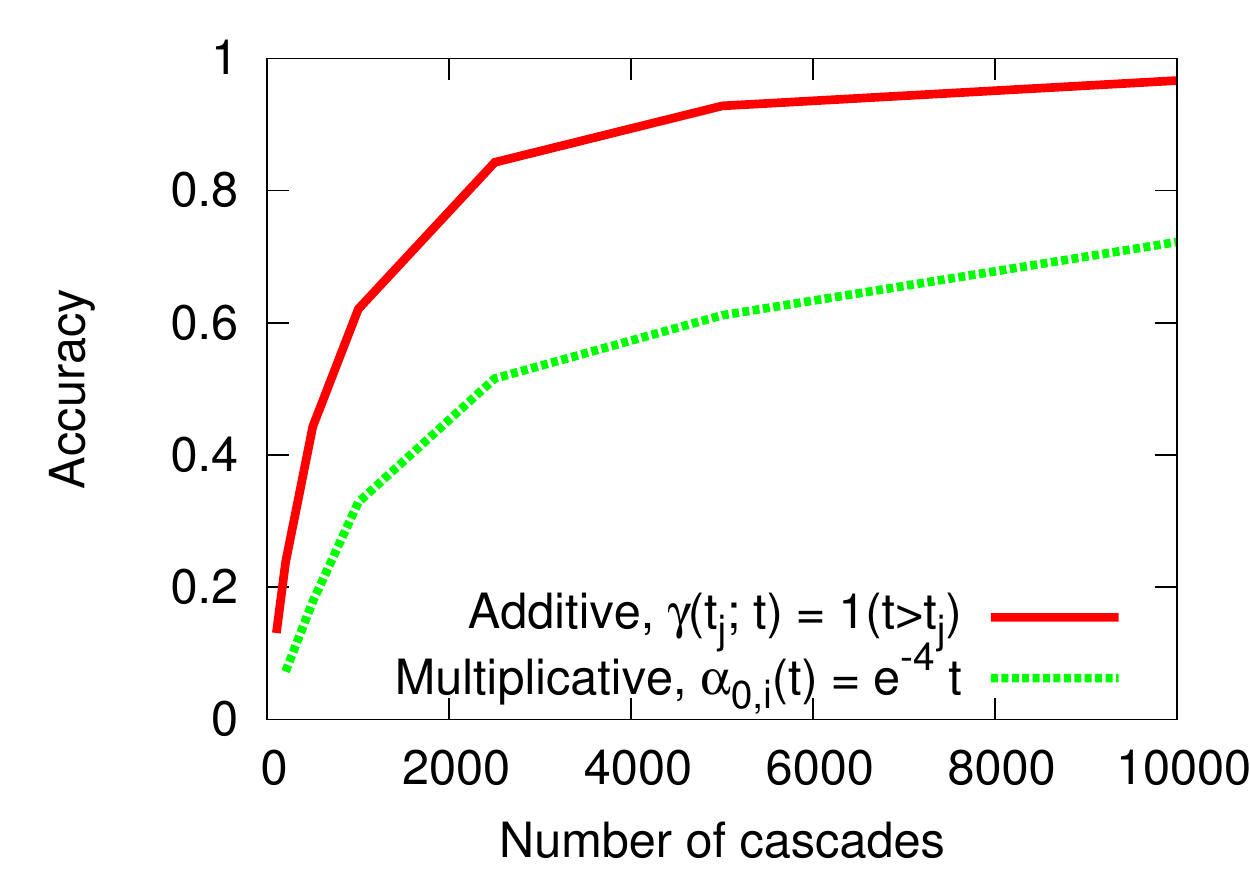}
	\label{fig:acc-hi-vs-num-cascades}}
	\\ \vspace{-2mm}
	\subfigure[MSE (C-P)]{\includegraphics[width=0.22\textwidth]{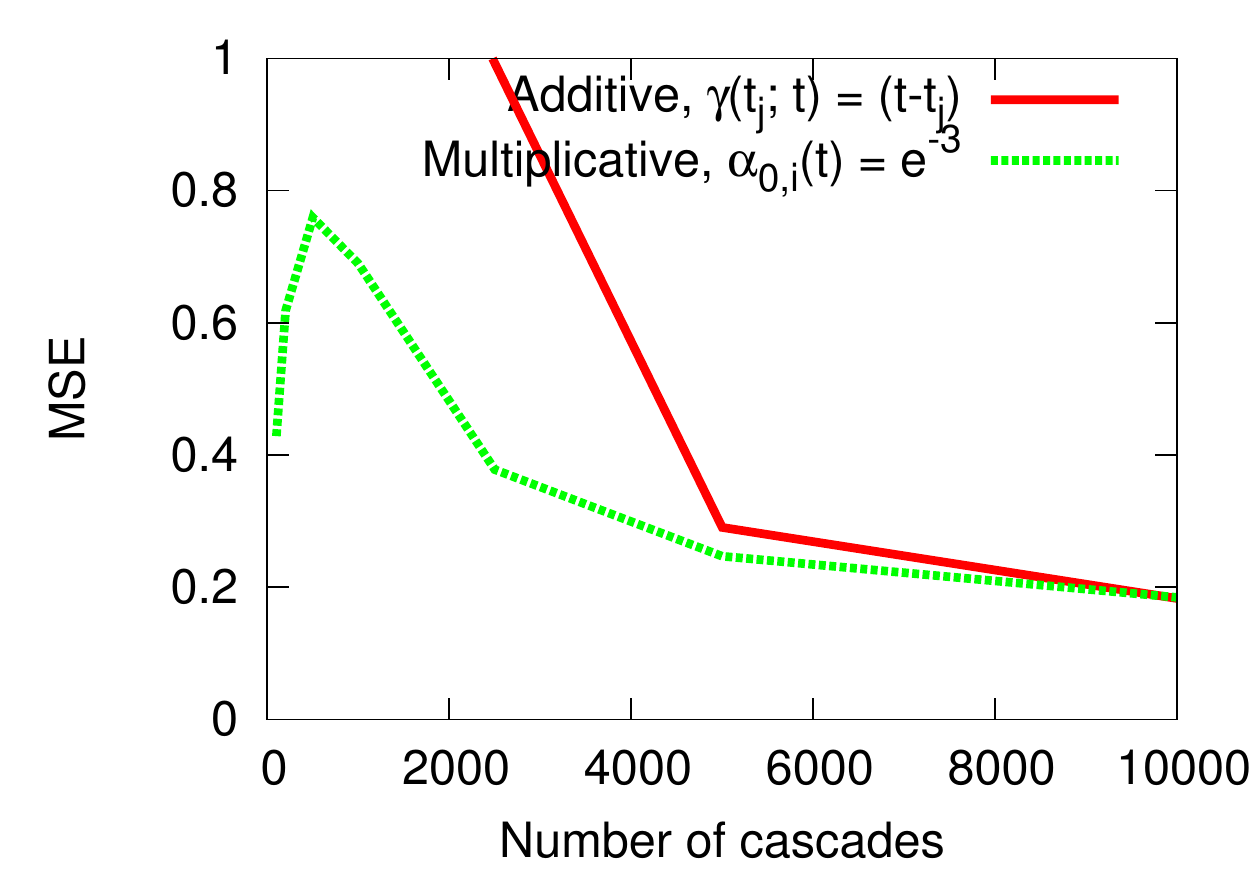}
	\label{fig:mse-cp-vs-num-cascades}}
	\subfigure[MSE (HI)]{\includegraphics[width=0.22\textwidth]{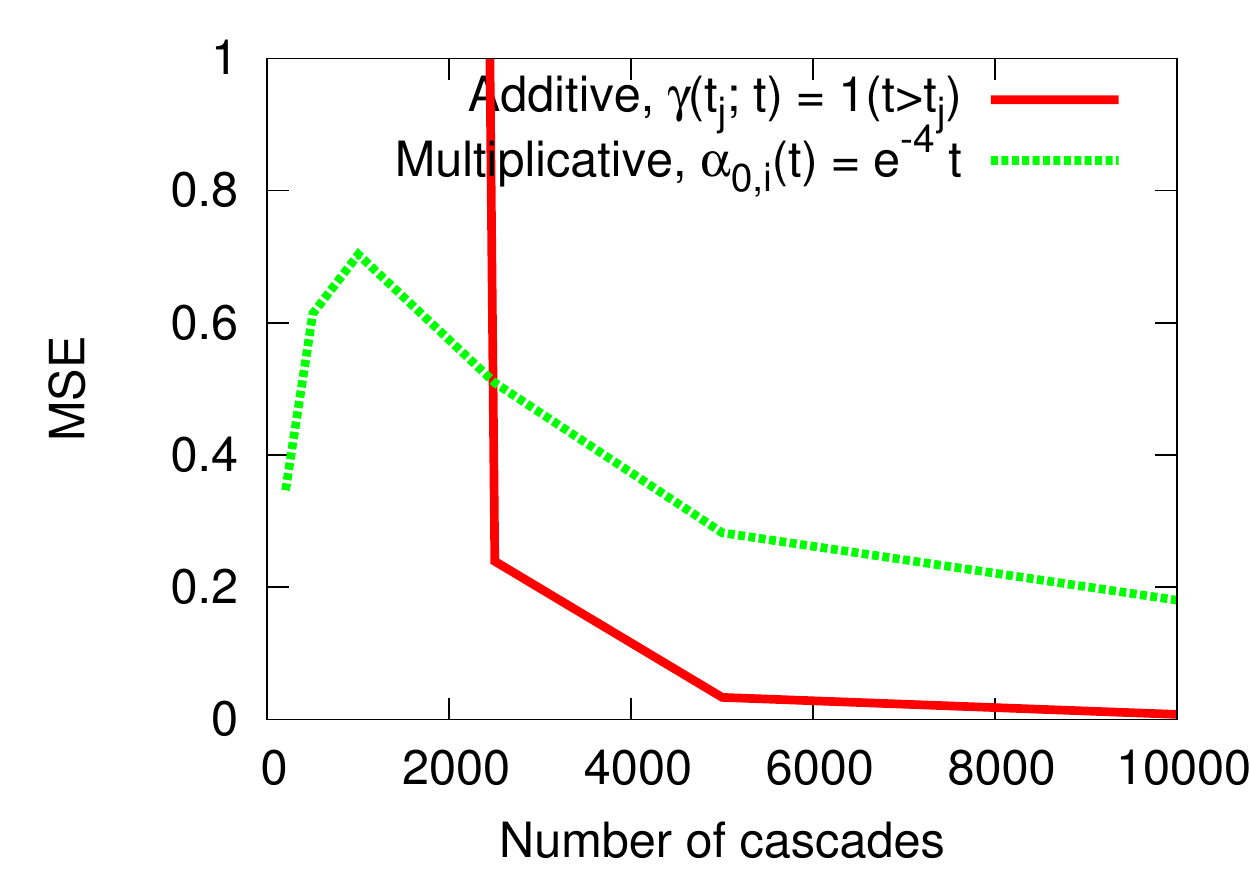}
	\label{fig:mse-hi-vs-num-cascades}}
	\vspace{-3mm}
 	\caption{
 	Edge accuracy and MSE of our inference me\-thods for additive and multiplicative propagation models against number of cascades. We used 
	1,024 node Core-Periphery (C-P) and Hierarchical (HI) Kronecker networks with an average of four edges per node and $T = 4$.
	} \label{fig:performance-vs-num-cascades}
\end{figure}

\subsection{Experiments on synthetic data} \label{sec:synthetic-experiments}
In this section, we compare the performance of our inference algorithms for additive and multiplicative models for different network structures, time shaping functions, baselines, and 
observation windows. 
We skip a comparison to other methods such as \netrate, \kernelcascade, \monet or \infopath since our model is able to mimic these methods by simply choosing the appropriate time shaping 
function $\gamma(\cdot ; \cdot)$, Table~\ref{tab:mapping}. Rather, we focus on comparing multiplicative and additive models.

\begin{figure}[t]
	\centering
	\subfigure[Edge accuracy]{\includegraphics[width=0.22\textwidth]{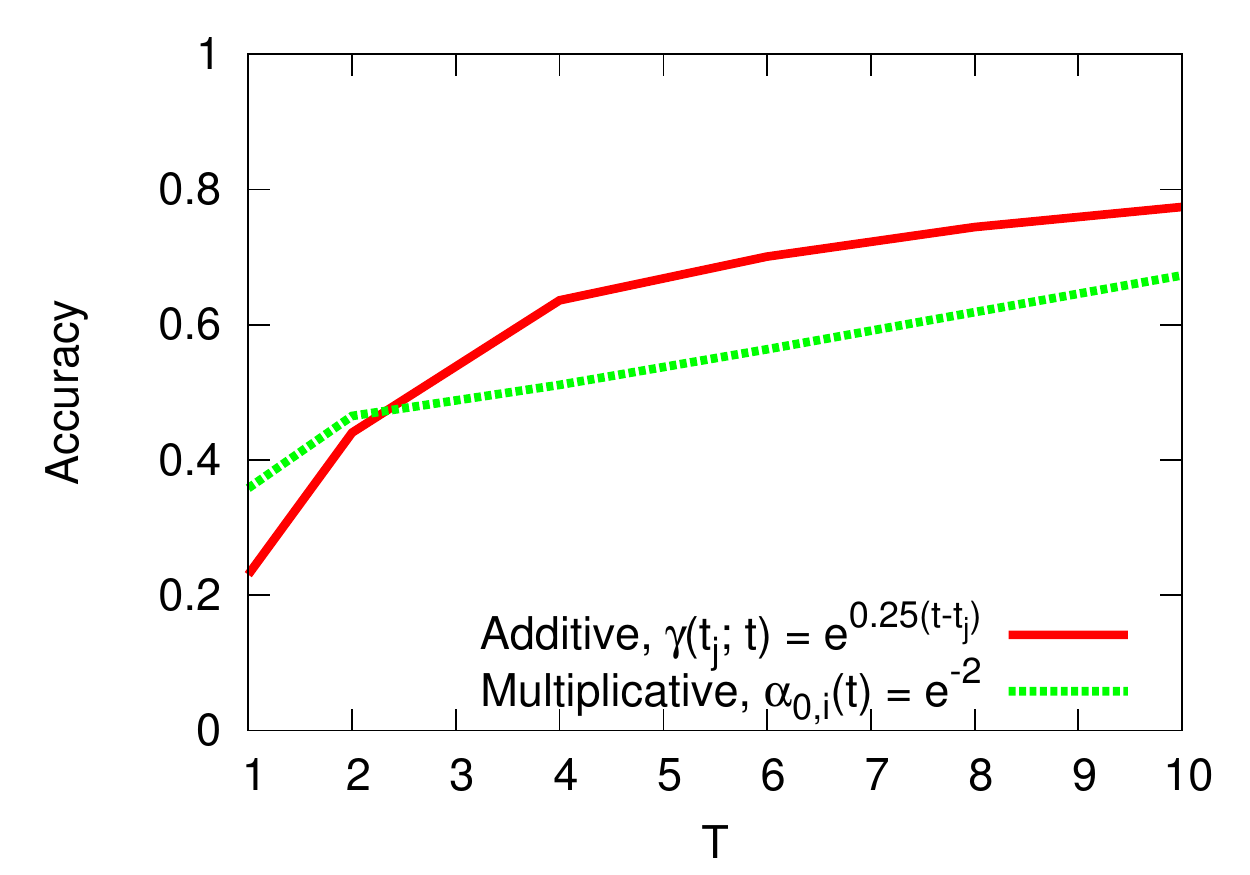}
	\label{fig:acc-random-vs-T}}
	\subfigure[MSE]{\includegraphics[width=0.22\textwidth]{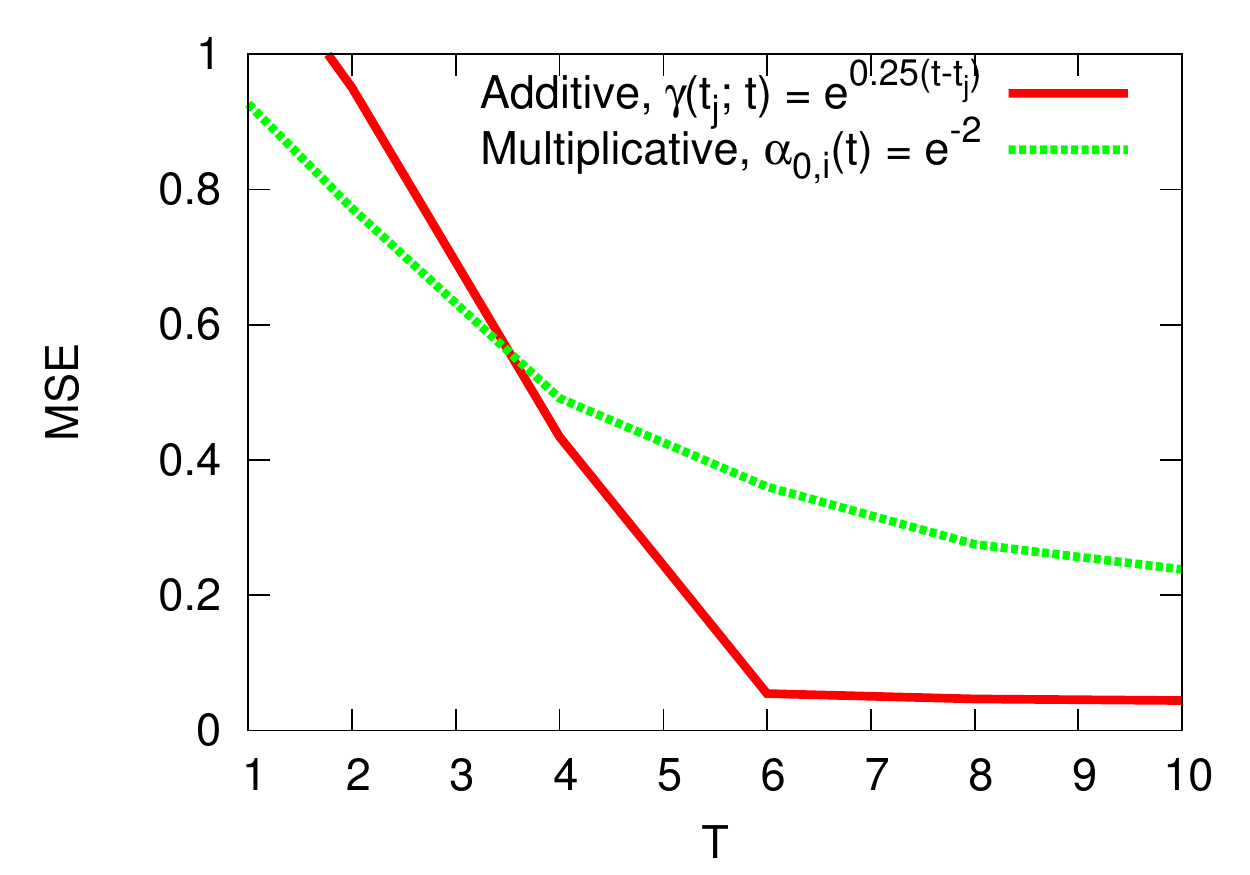}
	\label{fig:mse-random-vs-T}}
	\vspace{-3mm}
	\caption{
	Edge accuracy and MSE of our inference me\-thods for additive and multiplicative propagation models against observation window. We used 
	a 1,024 node Random Kronecker network with an average of 4 edges per node.
	} \label{fig:performance-vs-T}
\end{figure}

\xhdr{Experimental setup} First we generate realistic synthetic networks using the Kronecker graph model~\citep{leskovec2010kronecker}, and set the edge ha\-zard function parameters $\alpha_{j,i}$ randomly, drawn 
from a uniform distribution.
We then simulate and record a set of cascades propagating over the network using the additive or the multiplicative model.
For each cascade we pick the cascade initiator node uniformly at random and generate the infection times following a similar procedure as in~\citet{austin2012generating}: We draw a uniform random variable 
per node, and then use inverse transform sampling~\cite{devroye1986non} to generate piecewise likelihoods of node infections. Note that every time a parent of node $i$ gets infected, we need to consider a new interval in the piecewise likelihood of infection of node $i$.

\xhdr{Performance vs. number of cascades} We eva\-luate our inference methods by computing two di\-ffe\-rent measures: edge accuracy and mean squared error\- (MSE). Edge accuracy quantifies the fraction of edges the method was able to infer correctly: $1-\frac{\sum_{i,j} |I(\alpha^*_{i,j})-I(\hat{\alpha}_{i,j})|}{\sum_{i,j}I(\alpha^*_{i,j}) + \sum_{i,j} I(\hat{\alpha}_{i,j})}$, where $I(\alpha)=1$ if $\alpha > 0$ and $I(\alpha)=0$ otherwise. 
The MSE quantifies the error in the estimates of parameters $\alpha$: 
$E\big[(\alpha^*-\hat{\alpha})^2\big]$, where $\alpha^*$ is the true parameter of the model and $\hat{\alpha}$ is the estimated parameter. 

Figure~\ref{fig:performance-vs-num-cascades} shows the edge accuracy and the MSE against cascade size for two types of Kronecker networks: hie\-rar\-chi\-cal and core-periphery, using different additive and multiplicative propagation models.  
Comparing additive and multiplicative models we find that in order to infer the networks to same accuracy the multiplicative model requires more data. This means it is more difficult to discover the network and fit the parameters for the multiplicative model than for the additive model. Moreover, estimating the value of the model parameters is considerably harder than simply discovering edges and therefore more cascades are needed for accurate estimates.

\xhdr{Performance vs. observation window length} Lengthening the observation window increases the number of observed infections and results in a more representative sample of the underlying dynamics. Therefore, it 
should intuitively result in more accurate estimates for both the additive and multiplicative models. Figure~\ref{fig:performance-vs-T} shows performance against di\-ffe\-rent observation window lengths for a random 
network~\citep{erdos60random}, using additive and multiplicative models over 1,000 cascades. The experimental results support the above intuition. However, given a sufficiently large observation window, 
increasing further the length of the window does not increase performance significantly, as observed in case of the additive model with exponential time shaping function.
\begin{table}[!!t]
    \small
    \begin{center}
    \begin{tabular*}{0.48\textwidth}{@{\extracolsep{\fill}} l c c}
	\toprule
  \textbf{Topic or news event $(\bf{Q})$} & \textbf{\# sites} & \textbf{\# memes} 
  \\
	\midrule
	Arab Spring & 950 & 17,975 \\
	Bailout & 1,127 & 36,863 \\
	Fukushima & 1,244 & 24,888\\
	Gaddafi & 1,068 & 38,166 \\
	Kate Middleton & 1,292 & 15,112 \\
	\bottomrule
    \end{tabular*}
    \end{center}
	\vspace{-4mm}
  \caption{Topic and news event statistics.}
  \label{tab:cascades}
  \vspace{-2mm}
\end{table}

\begin{figure*}[!t]
  \centering
  \small
  \subfigure[CS: Arab Spring]{\includegraphics[width=0.18\textwidth]{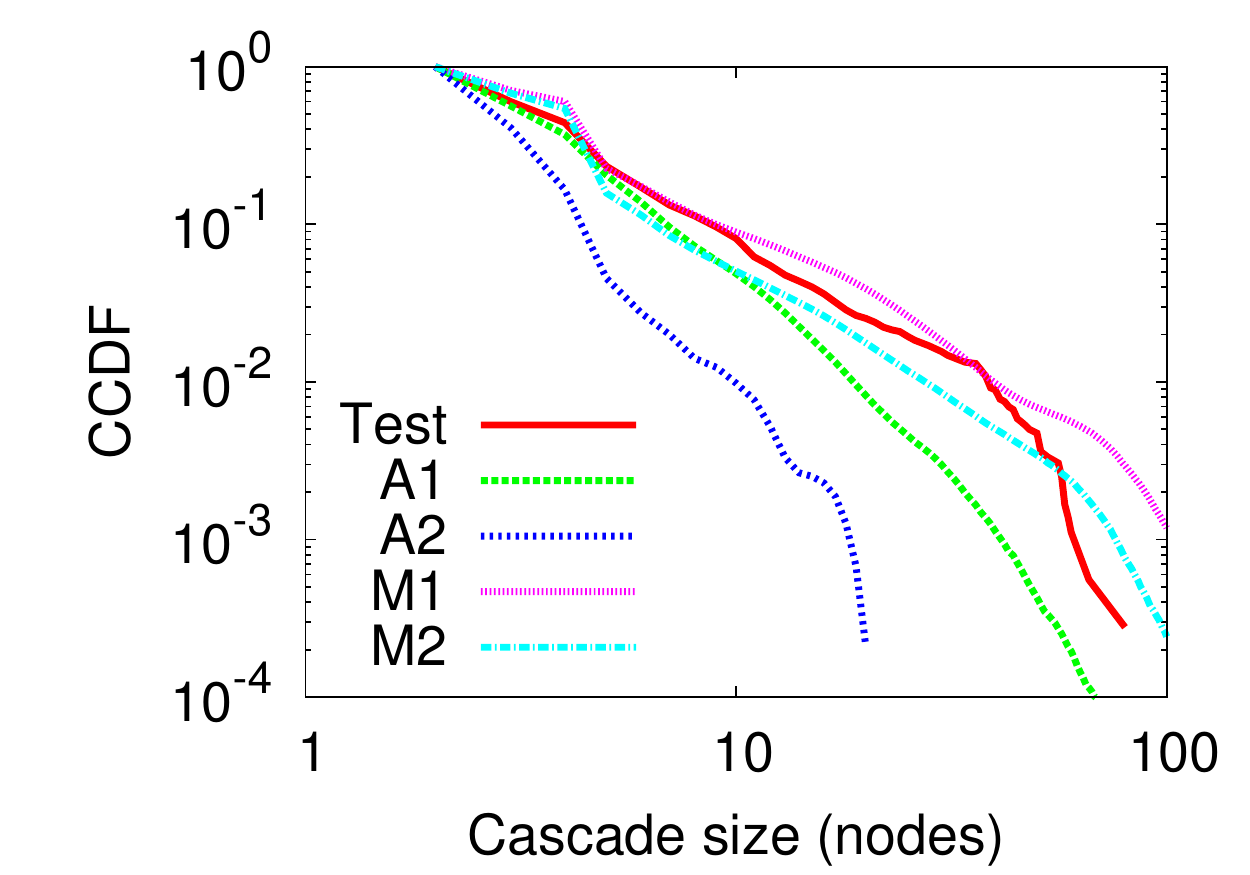} \label{fig:ttl-cascades-arab-spring}}
  \subfigure[CS: Bailout]{\includegraphics[width=0.18\textwidth]{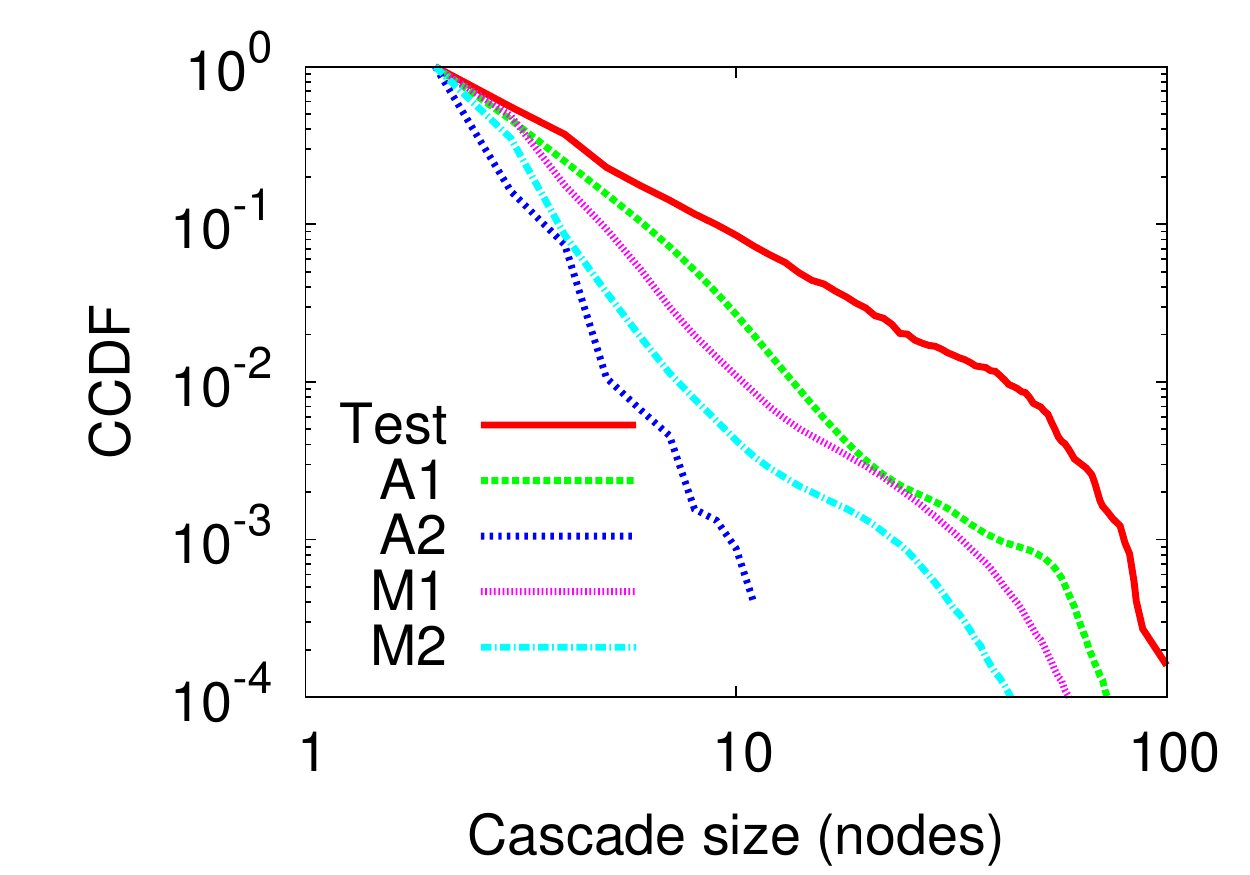} \label{fig:size-cascades-bailout}}
  \subfigure[CS: Fukushima]{\includegraphics[width=0.18\textwidth]{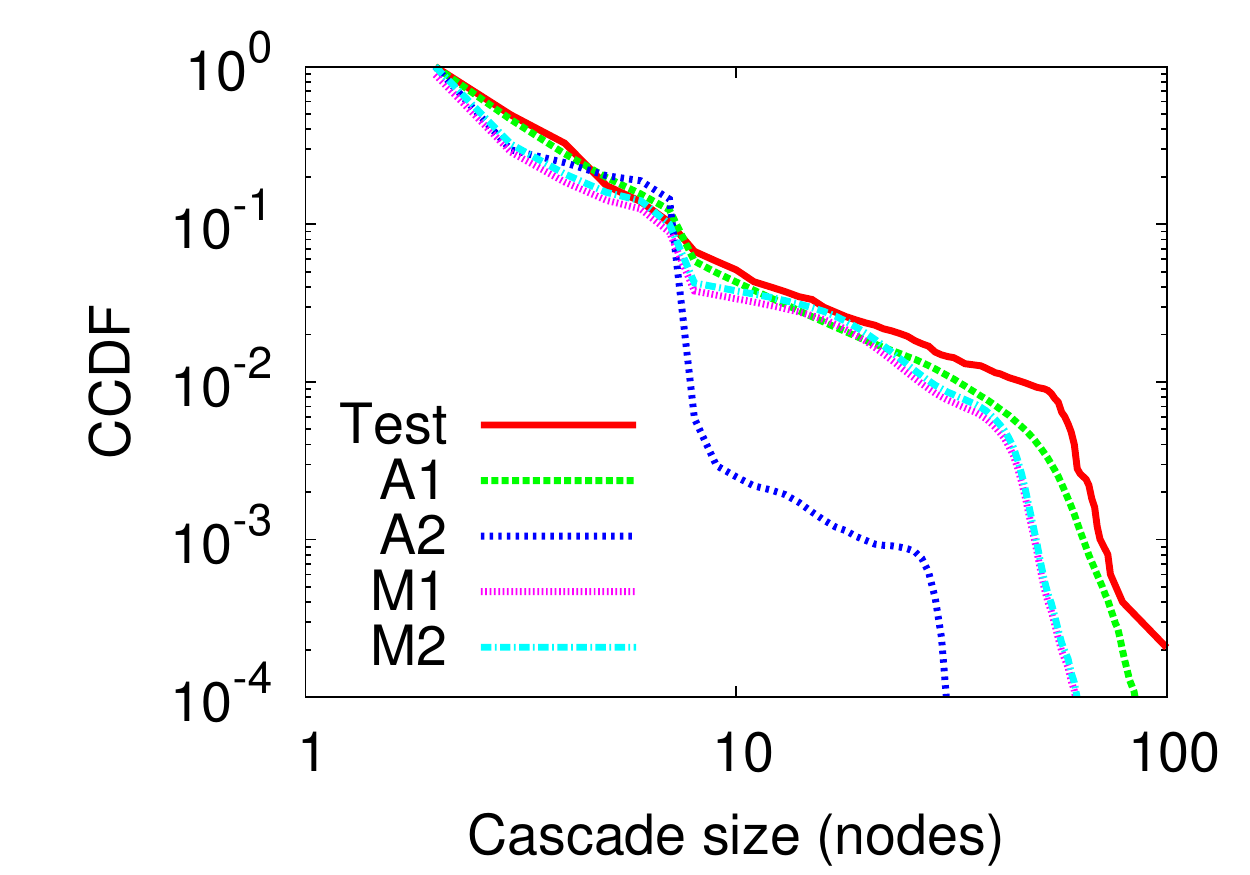} \label{fig:size-cascades-fukushima}}
  \subfigure[CS: Gaddafi]{\includegraphics[width=0.18\textwidth]{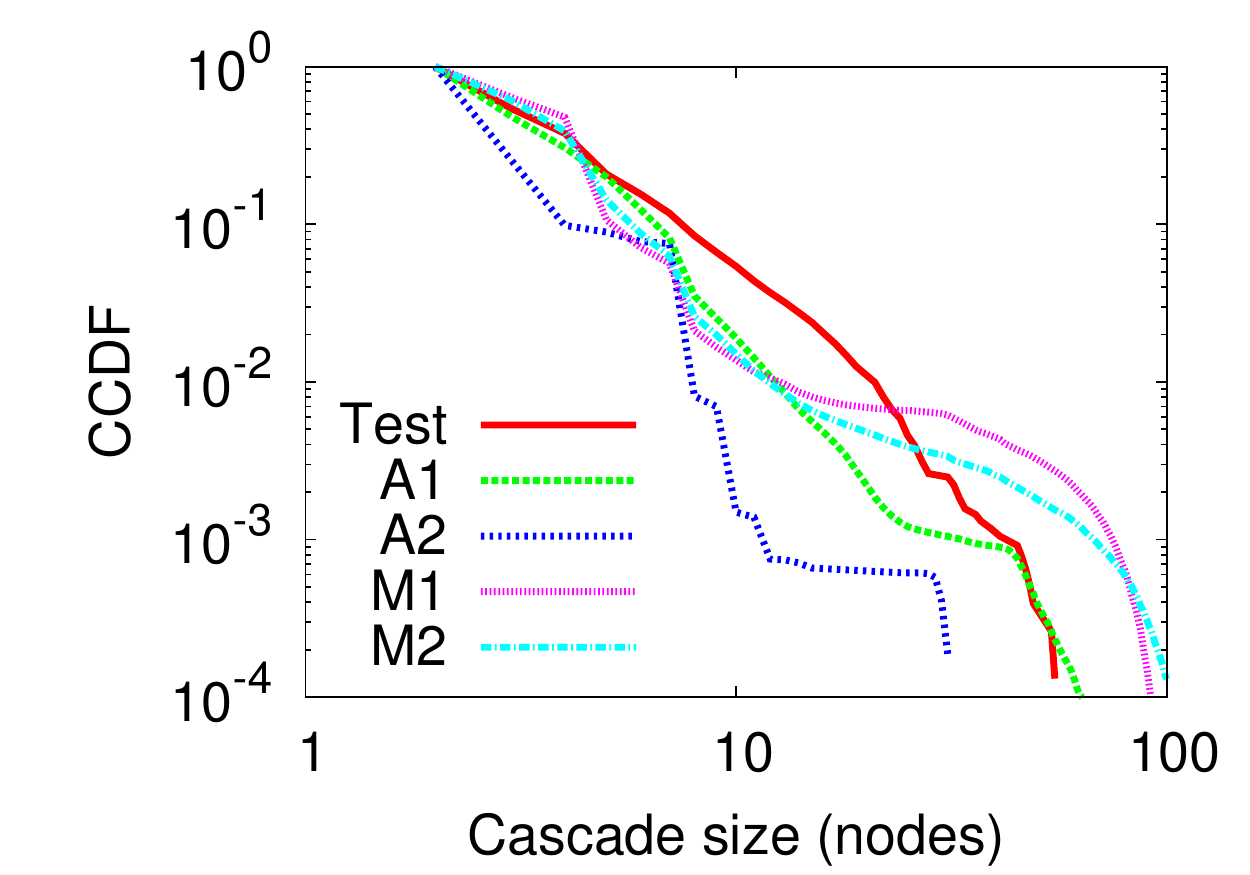} \label{fig:size-cascades-gaddafi}}
  \subfigure[CS: K. Middleton]{\includegraphics[width=0.18\textwidth]{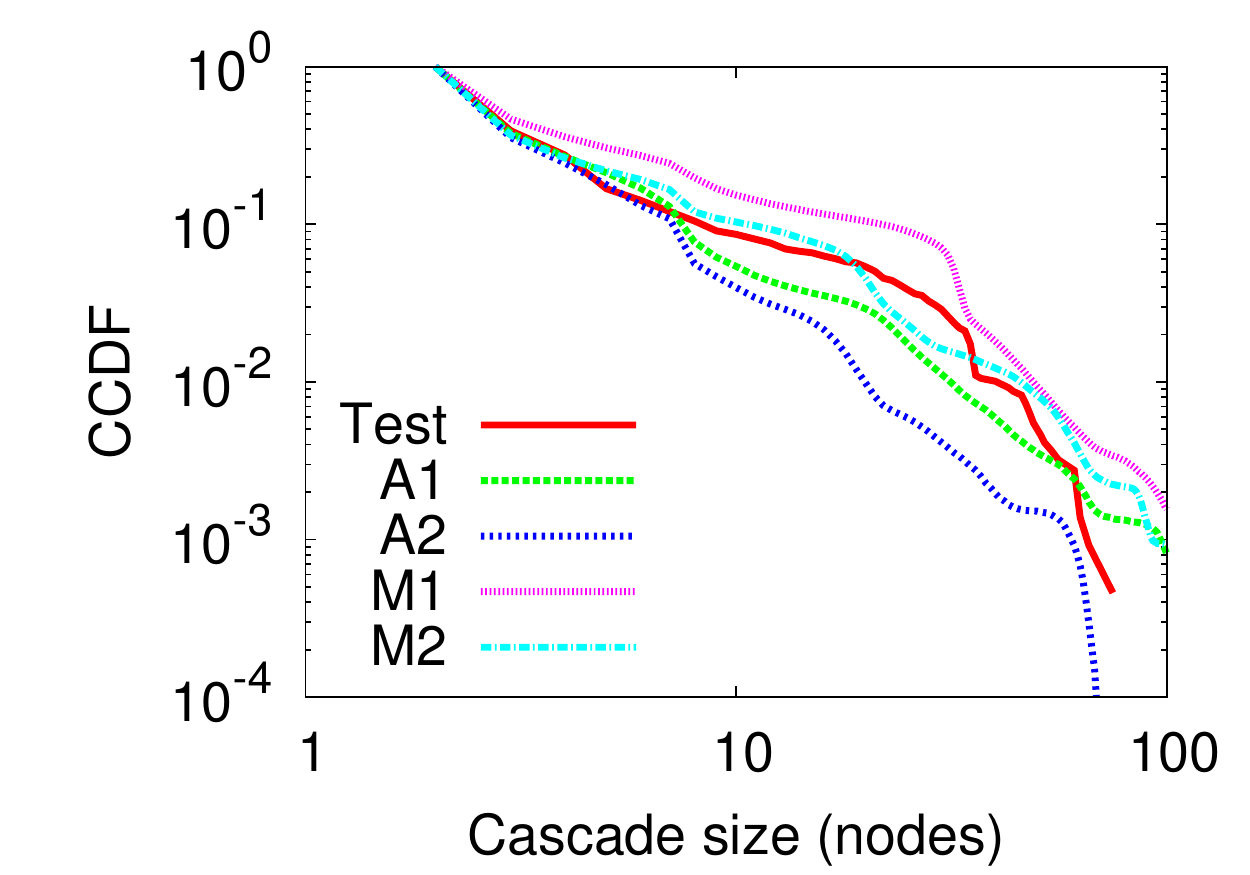} \label{fig:size-cascades-middleton}}\\ \vspace{-3mm}
  \subfigure[CD: Arab Spring]{\includegraphics[width=0.18\textwidth]{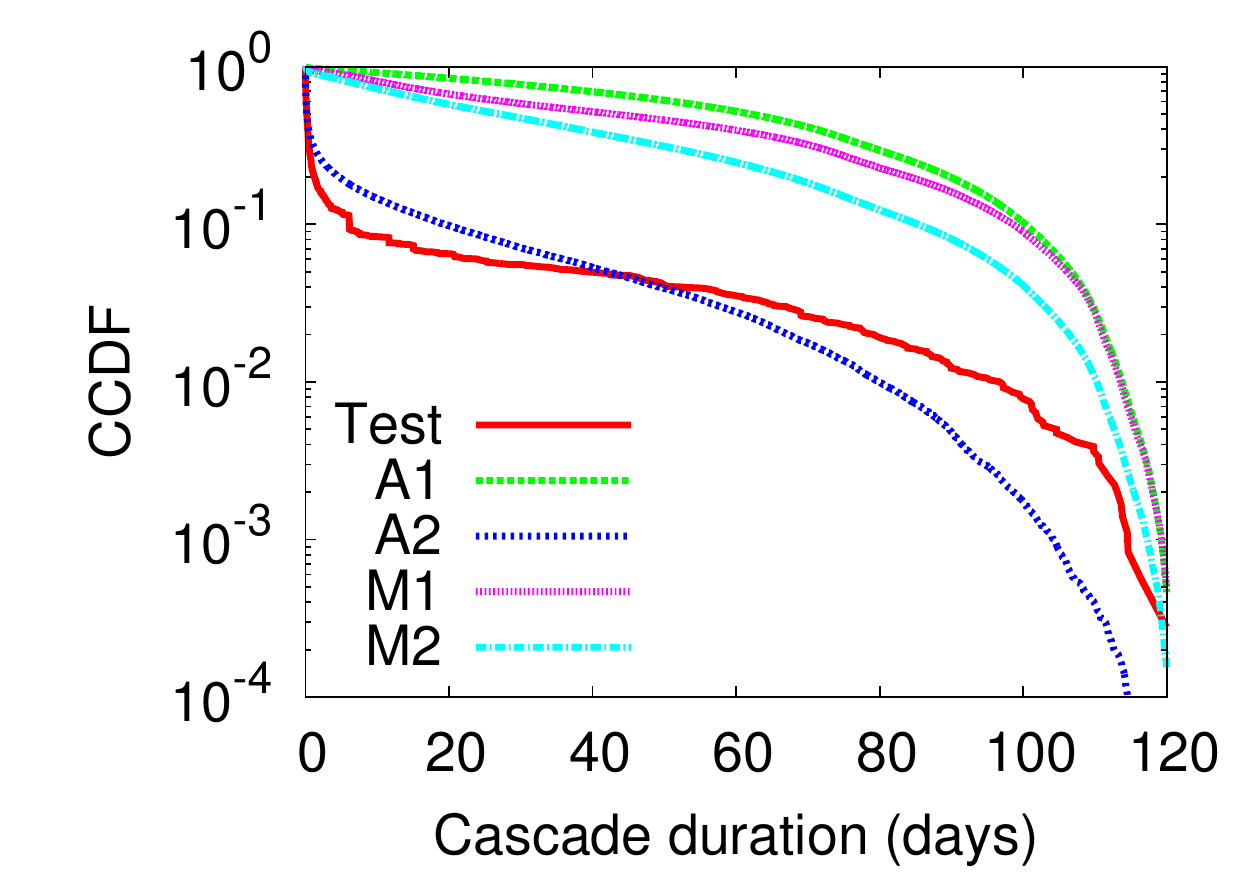} \label{fig:duration-cascades-arab-spring}}
  \subfigure[CD: Bailout]{\includegraphics[width=0.18\textwidth]{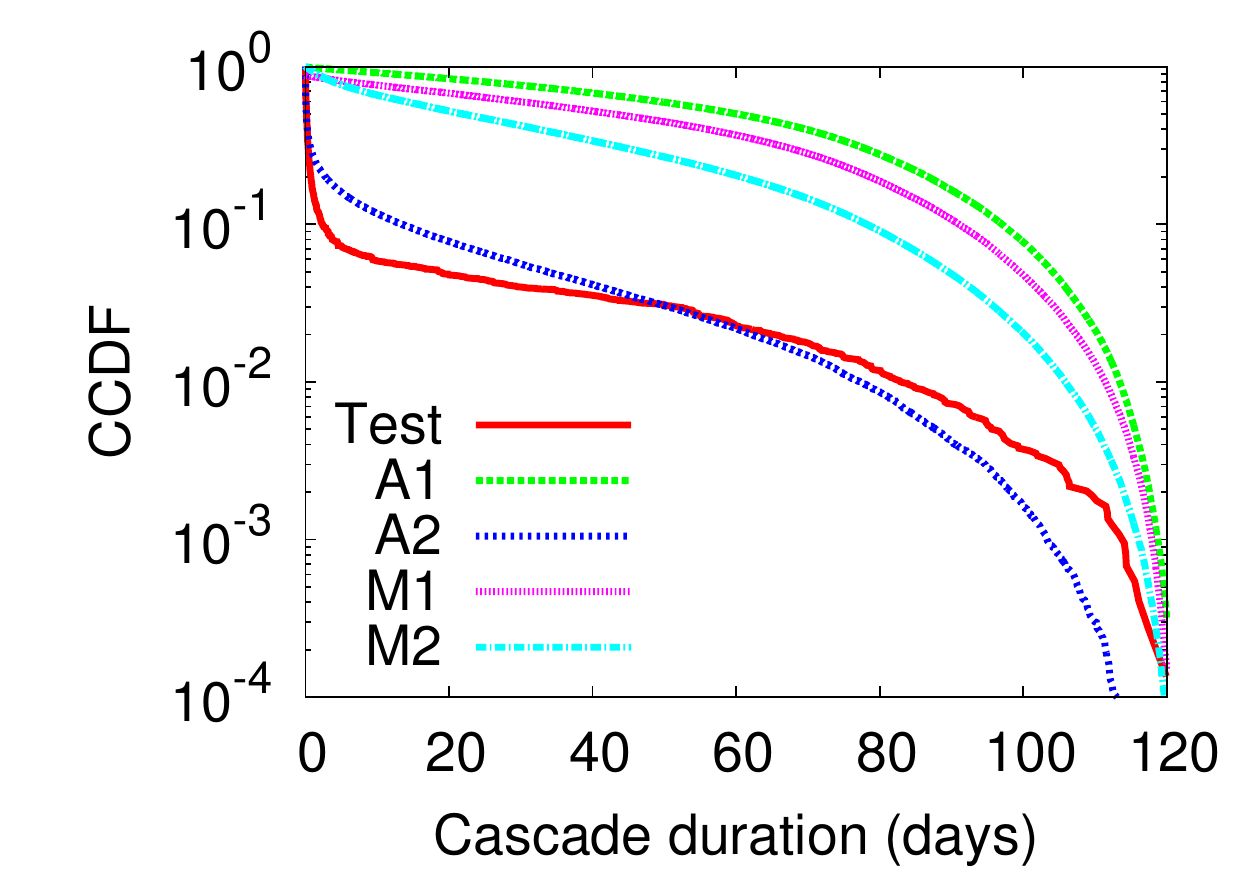} \label{fig:duration-cascades-bailout}}
  \subfigure[CD: Fukushima]{\includegraphics[width=0.18\textwidth]{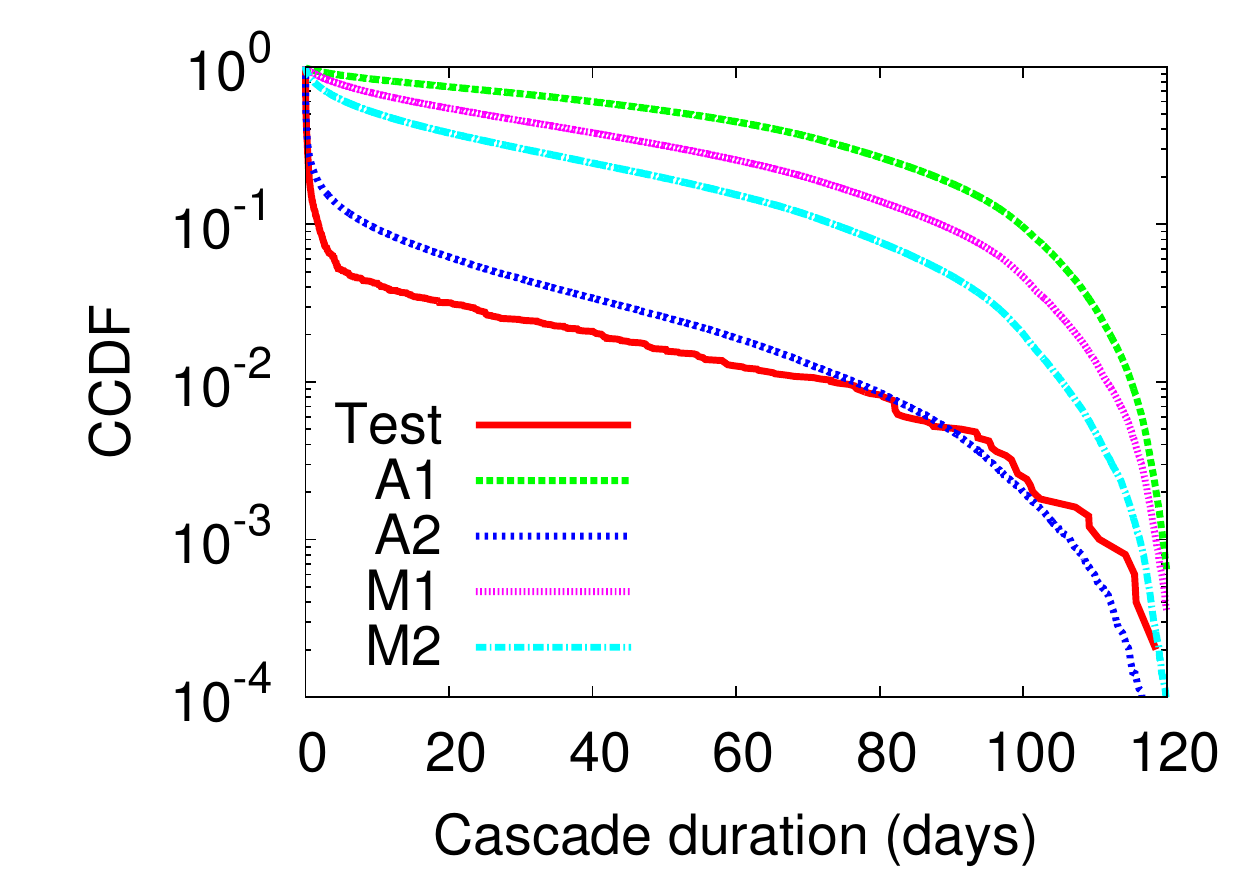} \label{fig:duration-cascades-fukushima}}
  \subfigure[CD: Gaddafi]{\includegraphics[width=0.18\textwidth]{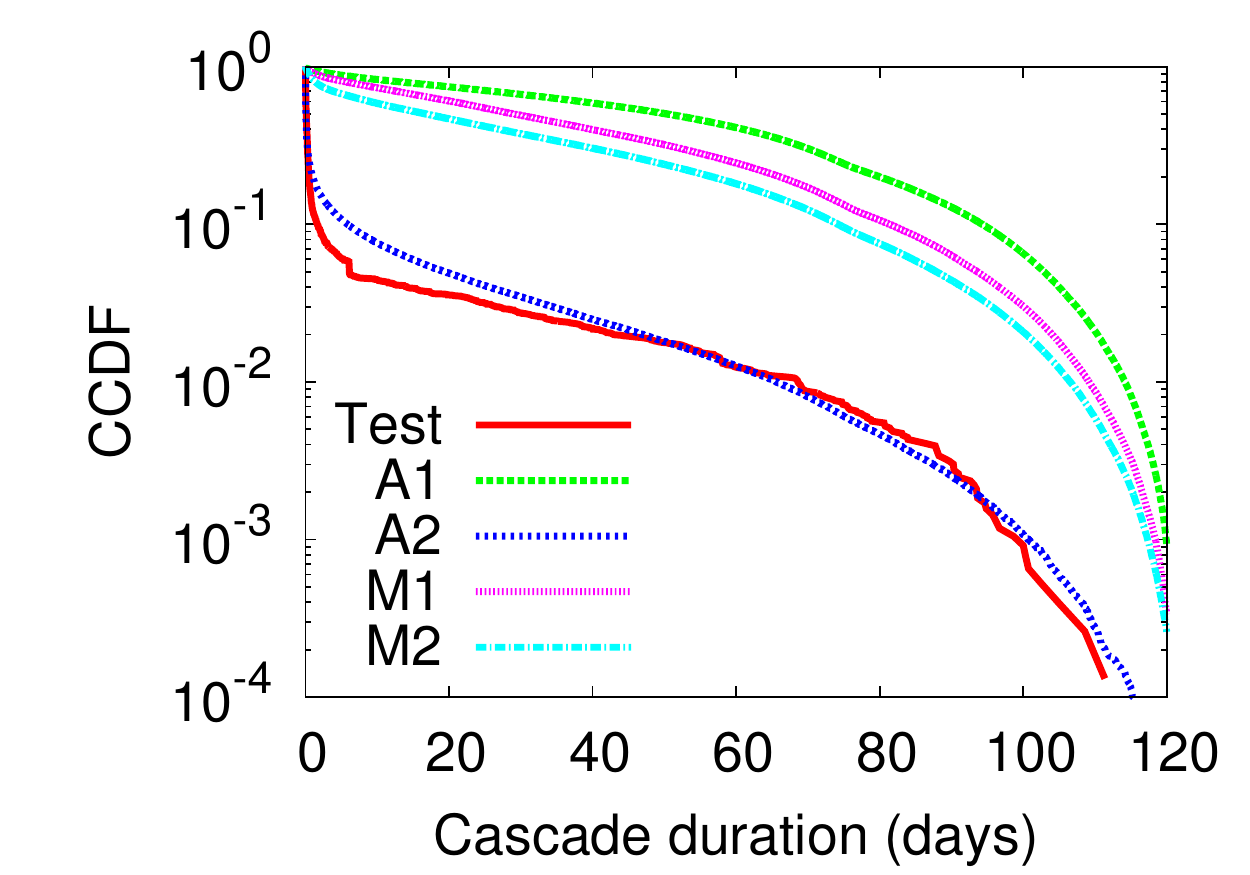} \label{fig:duration-cascades-gaddafi}}
  \subfigure[CD: K. Middleton]{\includegraphics[width=0.18\textwidth]{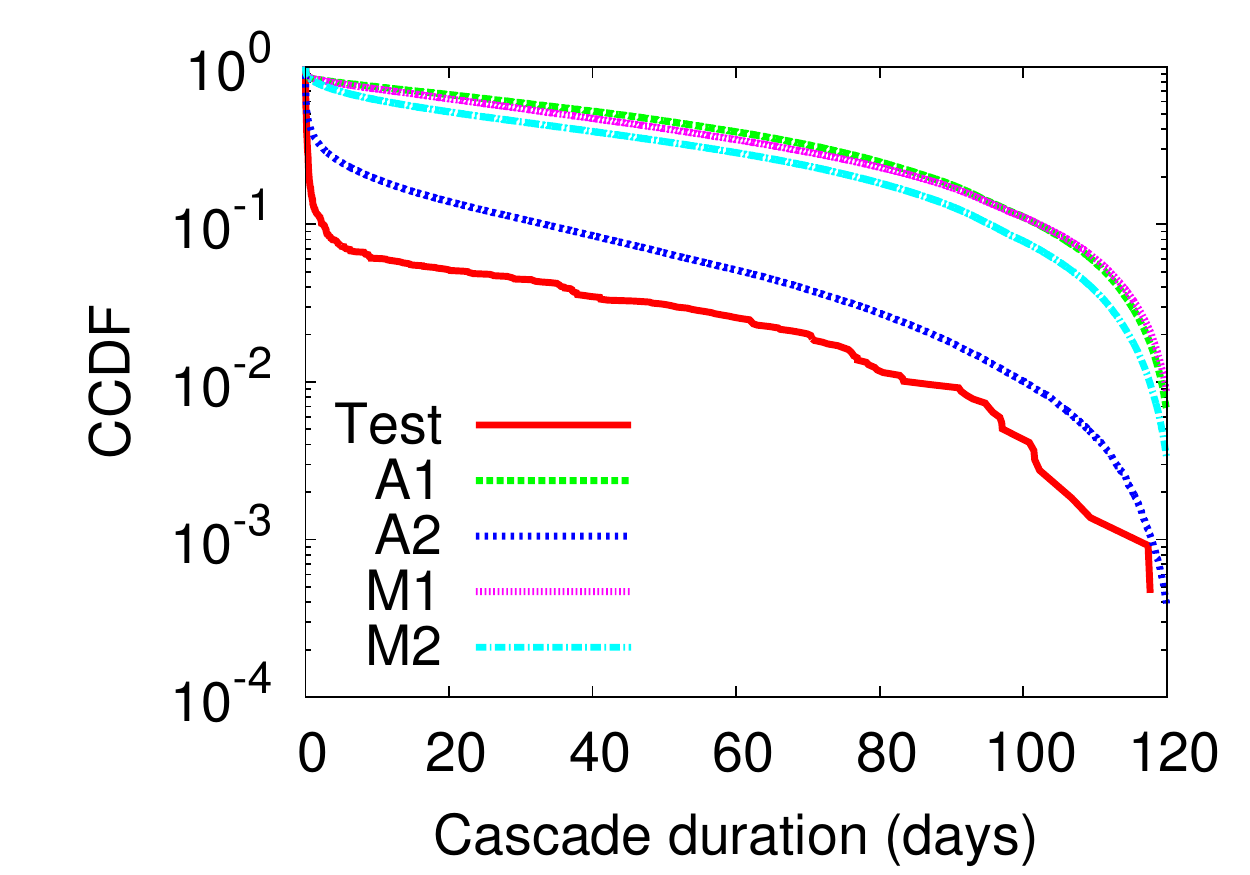} \label{fig:duration-cascades-middleton}} 
  \vspace{-3mm}
  \caption{Cascade size (CS) and cascade duration (CD) distributions for test sets on several topics. We compare the cascade test set (Test) against synthetic cascade sets, generated 
  using two additive models, one (A1) with $\gamma(t_j ; t_i) = I(t_i > t_j)$ and other (A2) with $\gamma(t_j ; t_i) = 1/(t_i - t_j)$, and two multiplicative models, one (M1) with 
  $\alpha_{0,i}(t_i) = e^{b}I(t_i > t_j)$ and another (M2) with $\alpha_{0,i}(t_i) = e^{b}/t_i$, where we set $b = -3$.
  }
\label{fig:size-duration-cascades}
\end{figure*}

\subsection{Experiments on real data} \label{sec:real-experiments}

\xhdr{Experimental setup} We trace the flow of information using \emph{memes}~\cite{leskovec2009kdd}. \emph{Memes} are short textual phrases (like, ``lipstick on a pig'') that 
travel through a set of blogs and mainstream media websites. We consider each meme $m$ as a separate information cascade $c_{m}$. Since all documents which contain memes are time-stamped, a cascade $c_{m}$ 
is simply a record of times when sites first mentioned meme $m$.
We use more than 10 million distinct memes from 3.3 million websites over a period of 4 months, from May 2011 till August 2011.

Our aim is to consider sites that actively spread memes over the Web, so we select the top 5,000 sites in terms of the number of memes they mentioned.
Moreover, we are interested in inferring propagation models related to particular topics or events. Therefore, we consider we are also given a keyword query $Q$ related to the event/topic of interest. When we infer the parameters of the additive or multiplicative models for a given query $Q$, we only consider documents (and the memes they mention) that include keywords in $Q$. Table~\ref{tab:cascades} summarizes the number of sites and meme cascades for several topics and real world events.

Unfortunately, true models (or ground truth) are unknown on real data. Previous network inference algorithms~\cite{manuel10netinf, manuel11icml, manuel13dynamic, 
lesong2012nips} have been validated using explicit hyperlink cascades. However, the use of hyperlinks to refer to the source of information is relatively 
rare (especially in mainstream media)~\cite{manuel10netinf} and hyperlinks only allow us to test for nodes which increase the instantaneous risk of infection of other nodes, as in the additive risk model, but not for nodes which either increase or decrease it.
To overcome this, we instead evaluate the predictive power of our models. For each query $Q$ we create a training and a test set of cascades. The training (test) set contains 80\% (20\%) of the recorded cascades for the event/topic of interest; both sets are disjoint and created at random. We use the training set to fit the parameters of our additive and multiplicative models and the test set to evaluate the models.

\xhdr{Cascade size prediction} For each event/topic of interest, we evaluate the predictive power of both mo\-dels, learned using the training set, by comparing the cascade size distribution of the test set against a synthetically generated cascade set using the trained models.
We build the synthetically generated cascade set by simulating a set of cascades starting from the true source nodes of the cascades in the test set using the model learned from the training set and an observation window equal to the one of the test set. Figures~\ref{fig:size-duration-cascades}(a-e) show the distribution of cascade sizes for the test sets and for the synthetically 
generated cascade sets using two different additive models and two multiplicative models. None of the models is a clear winner in terms of similarity with the test sets, but the additive model with inverse 
linear time shaping function (A2) tends to underestimate the cascade size. Surprisingly, the cascade size distributions in the synthetically generated cascade sets are very similar to the empirical distributions, specially up to 10 infected nodes per cascade.

\xhdr{Cascade duration prediction} Next we further evaluate the predictive power of both models, learned using the training sets, by comparing the cascade duration distribution of the test sets against synthetically generated cascade sets using the trained models. 
Figures~\ref{fig:size-duration-cascades}(f-j) show the distribution of the cascade duration of the test sets and the synthetically generated cascade sets using the same additive and multiplicative models. In this case, the performance of the models differs more dramatically. The additive model with inverse linear time shaping function gets the closest to the empirical distribution of the test 
set at the cost of underestimating the cascade size.

\section{Conclusion}
\label{sec:conclusions}
Our work here contributes towards a general mathematical theory of information propagation over networks while also providing flexible methods. Moreover, there are also many venues for future work. 
%
In the additive model, external influences that are endogenous to the network ~\cite{seth2012kdd} could be considered by including an extra additive term $\alpha_{i0}(t)$. In the multiplicative model, one could consider nonparametric baselines $\alpha_{i,0}(t)$ by fitting the model using partial likelihood.
Both models could be extended to include other types of covariates $\mathbf{s}(t)$ and also to consider time varying parameters $\alpha_{ji}(t)$ in order to infer dynamic networks~\cite{manuel13dynamic}.
Finally, developing goodness of fit tests would be useful to choose among models in a more principled manner.

\bibliographystyle{icml2013}
\bibliography{refs}

\end{document}